\documentclass[submission,copyright,creativecommons]{eptcs}
\providecommand{\event}{QPL 2021}

\usepackage[utf8]{inputenc}
\usepackage[english]{babel}
\usepackage[T1]{fontenc}

\usepackage{hyperref}
\usepackage{breakurl}

\usepackage{mathtools}
\usepackage{amsthm}
\usepackage{amsmath}
\usepackage{amssymb}
\usepackage{mathrsfs}
\usepackage{newtxtext, newtxmath}
\usepackage{tikz}
\usepackage{tikz-cd}
\usepackage{tikzfig}

\tikzstyle{node}=[minimum size=0.3cm]
\tikzstyle{Z}=[fill={rgb,255:red,230; green,254; blue,230}, draw={rgb,255: red,61; green,77; blue,61}, shape=circle]
\tikzstyle{X}=[fill={rgb,255:red,255; green,135; blue,136}, draw={rgb,255: red,102; green,54; blue,54}, shape=circle]
\tikzstyle{Y}=[fill=zxblue, draw=zxdblue, shape=circle]
\tikzstyle{Z_big}=[fill={rgb,255:red,230; green,254; blue,230}, draw={rgb,255: red,61; green,77; blue,61}, shape=circle, minimum width=1.6em, font={\small}]
\tikzstyle{X_big}=[fill={rgb,255:red,255; green,135; blue,136}, draw={rgb,255: red,102; green,54; blue,54}, shape=circle, minimum width=1.6em, font={\small}]
\tikzstyle{Y_big}=[fill=zxblue, draw=zxdblue, shape=circle, minimum width=1.6em, font={\small}]
\tikzstyle{Z_tri}=[fill={rgb,255:red,230; green,254; blue,230}, draw={rgb,255: red,61; green,77; blue,61}, regular polygon, regular polygon sides=3, draw, shape border rotate=0, inner sep=0pt, minimum width=15pt, line width=0.75]
\tikzstyle{X_tri}=[fill={rgb,255:red,255; green,135; blue,136}, draw={rgb,255: red,102; green,54; blue,54}, regular polygon, regular polygon sides=3, draw, shape border rotate=0, inner sep=0pt, minimum width=15pt, line width=0.75]
\tikzstyle{Z_tri_inv}=[fill={rgb,255:red,230; green,254; blue,230}, draw={rgb,255: red,61; green,77; blue,61}, regular polygon, regular polygon sides=3, draw, shape border rotate=180, inner sep=0pt, minimum width=15pt, line width=0.75, font={\footnotesize}]
\tikzstyle{X_tri_inv}=[fill={rgb,255:red,255; green,135; blue,136}, draw={rgb,255: red,102; green,54; blue,54}, regular polygon, regular polygon sides=3, draw, shape border rotate=180, inner sep=0pt, minimum width=15pt, line width=0.75]
\tikzstyle{Z_tri_l}=[fill={rgb,255:red,230; green,254; blue,230}, draw={rgb,255: red,61; green,77; blue,61}, regular polygon, regular polygon sides=3, draw, shape border rotate=90, inner sep=0pt, minimum width=15pt, line width=0.75]
\tikzstyle{X_tri_l}=[fill={rgb,255:red,255; green,135; blue,136}, draw={rgb,255: red,102; green,54; blue,54}, regular polygon, regular polygon sides=3, draw, shape border rotate=90, inner sep=0pt, minimum width=15pt, line width=0.75]
\tikzstyle{H}=[fill=yellow, draw=black, shape=rectangle, minimum width=2mm, minimum height=2mm]
\tikzstyle{Y_box}=[draw=black, shape=rectangle, minimum width=2mm, minimum height=2mm, tikzit fill={rgb,255: red,255; green,128; blue,0}]
\tikzstyle{Z_med}=[fill={rgb,255:red,230; green,254; blue,230}, draw={rgb,255: red,61; green,77; blue,61}, shape=circle, minimum width=1.1em, font={\footnotesize}]
\tikzstyle{X_med}=[fill={rgb,255:red,255; green,135; blue,136}, draw={rgb,255: red,102; green,54; blue,54}, shape=circle, minimum width=1.1em, font={\footnotesize}]
\tikzstyle{Y_med}=[fill=zxblue, draw=zxdblue, font={\footnotesize}, minimum width=1.1em, font={\footnotesize}, shape=circle]
\tikzstyle{ZP}=[fill={rgb,255:red,230; green,254; blue,230}, draw={rgb,255: red,61; green,77; blue,61}, regular polygon, regular polygon sides=3, shape border rotate=30]
\tikzstyle{ZM}=[fill={rgb,255:red,230; green,254; blue,230}, draw={rgb,255: red,61; green,77; blue,61}, regular polygon, regular polygon sides=3, shape border rotate=-30]
\tikzstyle{XP}=[fill={rgb,255:red,255; green,135; blue,136}, draw={rgb,255: red,102; green,54; blue,54}, regular polygon, regular polygon sides=3, shape border rotate=30]
\tikzstyle{XM}=[fill={rgb,255:red,255; green,135; blue,136}, draw={rgb,255: red,102; green,54; blue,54}, regular polygon, regular polygon sides=3, shape border rotate=-30]
\tikzstyle{small_box}=[fill=white, draw=black, shape=rectangle, minimum width=1.5cm, minimum height=1.5cm, font={\footnotesize}]
\tikzstyle{med_rectangle}=[fill=white, draw=black, shape=rectangle, minimum width=2.8cm, minimum height=1.5cm, font={\footnotesize}]
\tikzstyle{big_rectangle}=[fill=white, draw=black, shape=rectangle, minimum width=4.5cm, minimum height=1.5cm, font={\footnotesize}]
\tikzstyle{smol_box}=[fill=white, draw=black, shape=rectangle, minimum width=1cm, minimum height=1cm]
\tikzstyle{poo}=[minimum height=0.7cm, minimum width=0.7cm, path picture={\node at (path picture bounding box.center) {\includegraphics[width=0.7cm] {figures/poo}};}]
\tikzstyle{Z_long}=[fill={rgb,255:red,230; green,254; blue,230}, draw={rgb,255: red,61; green,77; blue,61}, shape=rectangle, rounded corners=0.25cm, minimum height=0.5cm, inner sep=0.25em, font={\scriptsize}]
\tikzstyle{X_long}=[fill={rgb,255:red,255; green,135; blue,136}, draw={rgb,255: red,102; green,54; blue,54}, shape=rectangle, rounded corners=0.25cm, minimum height=0.5cm, inner sep=0.25em, font={\scriptsize}]
\tikzstyle{med_rectv}=[fill=white, draw=black, shape=rectangle, minimum width=1.1cm, minimum height=2.4cm, font={\scriptsize}, inner sep=0.2cm]
\tikzstyle{Z dot}=[fill={rgb,255: red,0; green,127; blue,0}, draw=black, shape=circle, minimum width=1.5mm]
\tikzstyle{X dot}=[fill={rgb,255:red,255; green,21; blue,0}, draw=black, shape=circle, minimum width=1.5mm]
\tikzstyle{Z phase dot}=[fill={rgb,255: red,0; green,127; blue,0}, draw=black, shape=circle, font={\footnotesize}]
\tikzstyle{X phase dot}=[fill={rgb,255:red,255; green,21; blue,0}, draw=black, shape=circle, font={\footnotesize}]
\tikzstyle{ZYa}=[draw=black, shape=rectangle, rectangle split, rectangle split parts=2, rectangle split horizontal, rectangle split part fill={zxgreen, zxblue}, rectangle split draw splits=false, minimum height=2mm, font={\tiny}]
\tikzstyle{YZ}=[draw=black, shape=rectangle, rectangle split, rectangle split parts=2, rectangle split horizontal, rectangle split part fill={zxblue, zxgreen}, rectangle split draw splits=false, minimum height=2mm, font={\tiny}]
\tikzstyle{XYa}=[draw=black, shape=rectangle, rectangle split, rectangle split parts=2, rectangle split horizontal, rectangle split part fill={zxred, zxblue}, rectangle split draw splits=false, minimum height=2mm, font={\tiny}]
\tikzstyle{YX}=[draw=black, shape=rectangle, rectangle split, rectangle split parts=2, rectangle split horizontal, rectangle split part fill={zxblue, zxred}, rectangle split draw splits=false, minimum height=2mm, font={\tiny}]
\tikzstyle{tiny_box}=[fill=white, draw=black, shape=rectangle, minimum width=1cm, minimum height=1cm, font={\footnotesize}]
\tikzstyle{scalar}=[fill=white, draw=black, shape=diamond, font={\scriptsize}]

\tikzstyle{dashs}=[-, dashed, line width=0.15mm]
\tikzstyle{thick}=[-, line width=0.5mm]
\tikzstyle{arrow}=[->]
\tikzstyle{invisible}=[-, draw=none]
\tikzstyle{functor}=[-, fill={rgb,255: red,240; green,240; blue,240}]
\tikzstyle{boxedge}=[-, fill=white]

\usepackage{array}

\long\def\/*#1*/{}

\usepackage{pgfplots}
\pgfplotsset{compat=1.15}

\newcommand{\bra}[1]{\langle#1|}
\newcommand{\ket}[1]{|#1\rangle}

\definecolor{zxgreen}{RGB}{230,254,230}
\definecolor{zxred}{RGB}{255,135,136}
\definecolor{zxblue}{RGB}{116,116,235}

\definecolor{zxdgreen}{RGB}{91,107,91}
\definecolor{zxdred}{RGB}{142,94,94}
\definecolor{zxdblue}{RGB}{61,61,77}

\newcommand{\s}{\enspace}

\newcommand{\bb}[1]{\mathbb{#1}}
\renewcommand{\phi}{\varphi}

\newcommand{\B}{\bb{B}}
\newcommand{\N}{\bb{N}}

\newcommand{\R}{\bb{R}}
\newcommand{\F}{\bb{F}}
\newcommand{\C}{\bb{C}}
\renewcommand{\S}{\bb{S}}
\newcommand{\D}{\bb{D}}

\def\then{\mathbin{\raise 0.6ex\hbox{\oalign{\hfil$\scriptscriptstyle      \mathrm{o}$\hfil\cr\hfil$\scriptscriptstyle\mathrm{9}$\hfil}}}}

\newcommand{\tensor}{\otimes}
\newcommand{\id}{\mathtt{id}}
\newcommand{\dom}{\mathtt{dom}}
\newcommand{\cod}{\mathtt{cod}}

\newcommand{\Mat}{\mathbf{Mat}}

\newtheorem{definition}{Definition}[section]
\newtheorem{remark}[definition]{Remark}
\newtheorem{example}[definition]{Example}
\newtheorem{proposition}[definition]{Proposition}
\newtheorem{theorem}[definition]{Theorem}

\newtheorem{lemma}[definition]{Lemma}
\newtheorem{corollary}[definition]{Corollary}

\title{Diagrammatic Differentiation\\for Quantum Machine Learning}
\author{
Alexis Toumi$^{\star \dagger}$,
Richie Yeung$^\dagger$,
Giovanni de Felice$^{\star \dagger}$
\\
\institute{
$\star$ Department of Computer Science, University of Oxford \hspace{20pt}
$\dagger$ Cambridge Quantum Computing Ltd.}}

\def\titlerunning{Diagrammatic Differentiation
for Quantum Machine Learning}
\def\authorrunning{A. Toumi, R. Yeung \& G. de Felice}

\begin{document}
\maketitle


\begin{abstract}
We introduce diagrammatic differentiation for tensor calculus by generalising the
dual number construction from rigs to monoidal categories. Applying this to ZX
diagrams, we show how to calculate diagrammatically the gradient of a linear map
with respect to a phase parameter. For diagrams of parametrised quantum circuits,
we get the well-known parameter-shift rule at the basis of many variational
quantum algorithms. We then extend our method to the automatic
differentation of hybrid classical-quantum circuits, using diagrams with bubbles
to encode arbitrary non-linear operators. Moreover, diagrammatic differentiation
comes with an open-source implementation in DisCoPy, the Python library for
monoidal categories.
Diagrammatic gradients of classical-quantum circuits can then be simplified
using the PyZX library and executed on quantum hardware via the
tket compiler. This opens the door to many practical applications
harnessing both the structure of string diagrams and the computational power of
quantum machine learning.
\end{abstract}

\section*{Introduction}

String diagrams are a graphical language introduced by Penrose \cite{Penrose71}
to manipulate tensor expressions: wires represent vector spaces, nodes represent
multi-linear maps between them. In \cite{PenroseRindler84}, these diagrams are
used to describe the geometry of space-time and an extra piece of notation is
introduced: the covariant derivative is represented as a bubble around the tensor
to be differentiated. Joyal and Street \cite{JoyalStreet88,JoyalStreet91}
characterised string diagrams as the arrows of free monoidal categories, however
their geometry of tensor calculus makes no mention of differential calculus, it
only deals with composition and tensor.

In categorical quantum mechanics \cite{AbramskyCoecke08} string diagrams
are used to axiomatise quantum theory in terms of dagger compact-closed
categories. This culminated in the ZX-calculus \cite{CoeckeDuncan08},
a graphical language that provides a complete set of rules for qubit quantum
computing \cite{JeandelEtAl18a,HadzihasanovicEtAl18}. ZX diagrams
have recently been used for state-of-the-art quantum circuit optimisation
\cite{KissingerVanDeWetering20,DuncanEtAl20,DeBeaudrapEtAl20}, compilation
\cite{CowtanEtAl20,DeGriendDuncan20}, extraction \cite{BackensEtAl20} and error
correction \cite{ChancellorEtAl18,GidneyFowler19}. In recent work, ZX diagrams
have been used to study quantum machine learning \cite{Yeung20,ZhaoGao21} and
its application to quantum natural language processing
\cite{MeichanetzidisEtAl20a,CoeckeEtAl20}.

In this work, we introduce diagrammatic differentiation: a graphical notation
for manipulating tensor derivatives. On the theoretical side, we generalise
the dual number construction (discussed in section~\ref{1-dual-numbers})
from rigs to monoidal categories (section~\ref{2-dual-diagrams}). We then apply
this construction to the category of ZX diagrams (section~\ref{2b-differentiating-zx})
and of quantum circuits (section~\ref{3-dual-circuits}). In section~\ref{4-bubbles}
we give a formal definition of diagrams with bubbles and their gradient with the
chain rule. We use this to differentiate quantum circuits with neural
networks as classical post-processing. The theory comes with an
implementation in DisCoPy \cite{DeFeliceEtAl20}, the Python library for
monoidal categories. The gradients of classical-quantum circuits can then
be simplified using the PyZX library \cite{KissingerVanDeWetering19} and compiled
on quantum hardware via the tket compiler \cite{SivarajahEtAl20}.

\section*{Related work}

The same bubble notation for vector calculus is proposed in \cite{KimEtAl20},
but they have mainly pedagogical motivations and restrict themselves to the case
of three-dimensional Euclidean space. To the best of our knowledge, our
definition is the first formal account of string diagrams with bubbles for
tensor derivatives.

Differential categories \cite{BluteEtAl06} have been introduced to axiomatise
the notion of derivative. More recently reverse derivative categories
\cite{CockettEtAl19} generalised the notion of back-propagation, they have been
proposed as a categorical foundation for gradient-based learning
\cite{CruttwellEtAl21}. These frameworks all define the derivative
of a morphism with respect to its domain. In our setup however, we define
the derivative of parametrised morphism with respect to parameters that are
in some sense external to the category. Investigating the relationship between
these two definitions is left to future work.


\section{Dual numbers}\label{1-dual-numbers}

Dual numbers were first introduced by Clifford in 1873 \cite{Clifford73}.
Given a commutative rig (i.e. a riNg without Negatives) $\S$, the rig of dual
numbers $\D[\S]$ extends $\S$ by adjoining a new element $\epsilon$ such that $\epsilon^2 = 0$.
Concretely, elements of $\D[\S]$ are formal sums $s + s' \epsilon$ where
$s$ and $s'$ are scalars in $\S$.
We write $\pi_0, \pi_1 : \D[\S] \to \S$
for the projection on the real and epsilon component respectively.
Addition and multiplication of dual numbers are given by:
\begin{align} \begin{split}\label{linearity}
(a + a' \ \epsilon ) + (b + b' \ \epsilon)
\quad &= \quad (a + b) \s + \s (a + b') \ \epsilon
\end{split}\\
\begin{split}\label{product-rule}
(a + a' \ \epsilon ) \times (b + b' \ \epsilon)
\quad &= \quad (a \times b) \s + \s (a \times b' \ + \ a' \times b) \ \epsilon
\end{split}
\end{align}

A related notion is that of differential rig: a rig $\S$ equipped with a
derivation, i.e. a map $\partial : \S \to \S$ which preserves sums and satisfies
the Leibniz product rule
$\partial(f \times g) = f \times \partial(g) + \partial(f) \times g$ for all
$f, g \in \S$.
An equivalent condition is that the map $f \mapsto f + (\partial f) \epsilon$
is a homomorphism of rigs $\S \to \D[\S]$. The correspondance also works the
other way around: given a homorphism $\partial : \S \to \D[\S]$ such that
$\pi_0 \circ \partial = \id_\S$, projecting on the epsilon component is a
derivation $\pi_1 \circ \partial : \S \to \S$. The motivating example is the rig
of smooth functions $\S = \R \to \R$, where differentiation is a derivation.
Concretely, we can extend any smooth function $f : \R \to \R$ to a
function $f : \D[\R] \to \D[\R]$ over the dual numbers defined by:
\begin{equation}\label{dual-numbers-eq}
f(a + a' \epsilon) \quad = \quad f(a) \s + \s a' \times (\partial f)(a) \epsilon
\end{equation}

We can use equations~\ref{linearity}, \ref{product-rule} and \ref{dual-numbers-eq}
to derive the usual rules for gradients in terms of dual numbers.
For the identity function we have $\id(a + a' \epsilon) = \id(a) + a' \epsilon$,
i.e. $\partial \id = 1$. For the constant functions we have $c(a + a' \epsilon) =
c(a) + 0 \epsilon$, i.e. $\partial c = 0$.
For addition, multiplication and composition of functions, we can derive the
following \emph{linearity}, \emph{product} and \emph{chain} rules:
\begin{align} \begin{split}
    (f + g)(a + a' \epsilon)
    \s &= \s (f + g)(a) \s + \s a' \times (\partial f + \partial g)(a) \epsilon
\end{split}\\ \begin{split}
    (f \times g)(a + a' \epsilon)
    \s &= \s (f \times g)(a) \s + \s a' \times (f \times \partial g \ + \ \partial f \times g)(a) \epsilon
\end{split}\\ \begin{split}
    (f \circ g)(a + a' \epsilon)
    \s &= \s (f \circ g)(a) \s + \s a' \times (\partial g \ \times \ \partial f \circ g)(a) \epsilon
\end{split} \end{align}

This generalises to smooth functions $\R^n \to \R^m$, where the partial
derivative $\partial_i$ is a derivation for each $i < n$.
The functions $\F_2^n \to \F_2^m$ on the two-element
field $\F_2$ with elementwise XOR as sum and conjunction as product
also forms a differential rig.
The partial derivative is given by $(\partial_i f)(\vec{x}) =
f(\vec{x}_{[x_i \mapsto 0]}) \oplus f(\vec{x}_{[x_i \mapsto 1]})$.
Intuitively, the $\F_2$ gradient $\partial_i f(\vec{x}) \in \F_2^m$ encodes
which coordinates of $f(\vec{x})$ actually depend on the input $x_i$.
An example of differential rig that isn't also a ring is given by the set
$\N[X]$ of polynomials with natural number coefficients, again each
partial derivative is a derivation.

A more exotic example is the rig of Boolean functions with elementwise
disjunction as sum and conjunction as product. Boolean functions $\B^n \to \B^m$
can be represented as tuples of $m$ propositional formulae over $n$ variables.
The partial derivative $\partial_i$ for $i < n$ is defined by induction over
the formulae: for variables we have $\partial_i x_j = \delta_{ij}$,
for constants $\partial_i 0 = \partial_i 1 = 0$ and for negation
$\partial_i \neg \phi = \neg \partial_i \phi$. The derivative of disjunctions
and conjunctions are given by the linerarity and product rules.
Equivalently, the gradient of a propositional formula can be given by
$\partial_i \phi = \neg \phi_{[x_i \mapsto 0]} \land \phi_{[x_i \mapsto 1]}$.
Concretely, a model satisfies $\partial_i \phi$ if and only if it satisfies
$\phi \leftrightarrow x_i$: the derivative is true when the variable and the
formula are positively correlated. Substituting $x_i$ with its negation,
we get that a model satisfies $\partial_i \phi_{[x_i \mapsto \neg x_i]}$
if and only if it satisfies $\phi \leftrightarrow \neg x_i$, i.e. iff variable
and formula are anti-correlated.
Note that although $\B$ and $\F_2$ are isomorphic as sets, they are distinct
rigs. Their derivations are related however by $\partial^{\F_2}_i f \mapsto
\partial^\B_i \phi \lor \partial^\B_i \phi_{[x_i \mapsto \neg x_i]}$
for $\phi : \B^n \to \B$ the formula corresponding to the function
$f : \F_2^n \to \F_2$. That is, a Boolean function depends on an input variable
precisely when either the corresponding formula is positively correlated or
anti-correlated.

Dual numbers are a fundamental tool for \emph{automatic differentiation}
\cite{Hoffmann16}, i.e.
they allow to compute the derivative of a function automatically from its
definition. The key idea is that given a definition of $f : \S^n \to \S^m$ as a
composition of elementary functions, we can compute $(\partial_i f)(a)$ by
evaluating $f(a + \epsilon)$ and projecting on the epsilon component.


\section{Dual diagrams}\label{2-dual-diagrams}

Our main technical contribution is to generalise derivations from rigs to
monoidal categories with sums. Applying this to free monoidal categories,
where the arrows are string diagrams, we say a derivation is diagrammatic
when it commutes with the interpretation of the diagrams. We take two
different flavours of the ZX-calculus as our main examples.

Let $(\mathbf{C}, \otimes, 1)$ be a monoidal category with sums, i.e.
it has commutative monoids on each homset $(+) : \prod_{x,y}
\mathbf{C}(x, y) \times \mathbf{C}(x, y) \to \mathbf{C}(x, y)$
with unit $0 \in \prod_{x,y} \mathbf{C}(x, y)$
such that composition and tensor distribute over the sum.
Note that a one-object monoidal category with sums is simply a rig.
Our motivating example is the category $\mathbf{Mat}_\S$ with natural numbers
as objects and matrices valued in a commutative rig $\S$ as arrows, with matrix
multiplication as composition, Kronecker product as tensor and entrywise sum.
We define the category $\D[\mathbf{C}]$ by adjoining a scalar (i.e. an
endomorphism of the monoidal unit) $\epsilon$ such that $\epsilon \otimes \epsilon = 0$\footnote{
Note that in the case when $\mathbf{C}$ is not symmetric monoidal (or at least braided) the axiom $\epsilon \otimes f = f \otimes \epsilon$ is also needed.
}.
Concretely, the objects of $\D[\mathbf{C}]$ are the same as those of $\mathbf{C}$, the arrows
are given by formal sums $f + f' \epsilon$ of parallel arrows $f, f' \in \mathbf{C}$.
Composition and tensor are both given by the product rule:
\begin{align}
    (f + f' \epsilon) \ \then \ (g + g' \epsilon)
    &\s = \s f \ \then \ g \s + \s (f' \then g \ + \ f \then g') \ \epsilon\\
    (f + f' \epsilon) \tensor (g + g' \epsilon)
    &\s = \s f \tensor g \s + \s (f' \tensor g \ + \ f \tensor g') \ \epsilon
\end{align}

We say that a unary operator on homsets $\partial : \coprod_{x,y}
\mathbf{C}(x, y) \to \mathbf{C}(x, y)$ is a derivation whenever it satisfies the product rules for both composition
$\partial (f \then g) = (\partial f) \then g + f \then (\partial g)$
and tensor
$\partial (f \tensor g) = (\partial f) \tensor g + f \tensor (\partial g)$.
An equivalent condition is that the map $f \mapsto f + (\partial f) \epsilon$
is a sum-preserving monoidal functor $\mathbf{C} \to \D[\mathbf{C}]$.
Again, the correspondance between dual numbers and derivations works the other
way around: given a sum-preserving monoidal functor
$\partial : \mathbf{C} \to \D[\mathbf{C}]$ such that
$\pi_0 \circ \partial = \id_{\mathbf{C}}$, projecting on the epsilon component
gives a derivation $\pi_1 \circ \partial : \coprod_{x,y}
\mathbf{C}(x, y) \to \mathbf{C}(x, y)$. The following propositions characterise
the derivations on the category of matrices valued in a commutative rig $\S$.

\begin{proposition}
Dual matrices are matrices of dual numbers, i.e.
$\D[\mathbf{Mat}_\S] \simeq \mathbf{Mat}_{\D[\S]}$.
\end{proposition}

\begin{proof}
The isomorphism is given by
$\big( \sum_{ij} f_{ij} \ket{j} \bra{i} \big)
\ + \ \big( f'_{ij} \sum_{ij} \ket{j}  \bra{i} \big) \epsilon
\s \longleftrightarrow \s
\sum_{ij} (f_{ij} + f'_{ij} \epsilon) \ket{j} \bra{i}$.
\end{proof}

\begin{proposition}
Derivations on $\mathbf{Mat}_\S$ are in one-to-one correspondance with
derivations on $\S$.
\end{proposition}

\begin{proof}
A derivation on $\mathbf{Mat}_\S$ is uniquely determined by its action on
scalars in $\S$. Conversely, applying a derivation $\partial : \S \to \S$
entrywise on matrices yields a derivation on $\mathbf{Mat}_\S$.
\end{proof}

Fix a monoidal signature $\Sigma$ with objects $\Sigma_0$ and boxes $\Sigma_1$.
Let $\mathbf{C}_\Sigma$ be the free monoidal category it generates:
the objects are types, i.e. lists of generating objects
$t = t_1, \dots, t_n \in \Sigma_0^\star$, the arrows are
string diagrams with boxes in $\Sigma_1$.
Let $\mathbf{C}_\Sigma^+$ be the free monoidal category with sums:
the objects are also given by types, the arrows are formal sums, i.e.
bags\footnote{A bag of $X$, also called a multiset, is a function $X \to \N$.
Addition of bags is done pointwise with unit the constant zero.},
of string diagrams.
We assume our diagrams are interpreted as matrices, i.e. we fix a sum-preserving
monoidal functor $[\![-]\!]  : \mathbf{C}_\Sigma^+ \to \mathbf{Mat}_\S$
for $\S$ a commutative rig with a derivation $\partial : \S \to \S$.
Our main two examples are the standard ZX-calculus with smooth functions
$\R^n \to \R$ as phases and the algebraic ZX-calculus over $\S$,
introduced in \cite{Wang20}.

Applying the dual number construction to $\mathbf{C}_\Sigma^+$,
we get the category of dual diagrams $\D[\mathbf{C}_\Sigma^+]$ which is where
diagrammatic differentiation happens.
By the universal property of $\mathbf{C}_\Sigma^+$, every derivation
$\partial : \mathbf{C}_\Sigma^+ \to \D[\mathbf{C}_\Sigma^+]$ is uniquely
determined by its image on the generating boxes in $\Sigma_1$. Intuitively,
if we're given the derivative for each box, we can compute the derivative
for every sum of diagram using the product rule.
We say that the interpretation $[\![-]\!] : \mathbf{C}_\Sigma^+ \to \mathbf{Mat}_\S$
admits diagrammatic differentiation if there is a derivation $\partial$ on
$\mathbf{C}_\Sigma^+$ such that
$[\![-]\!] \circ \partial = \partial \circ [\![-]\!]$, i.e. the interpretation
of the gradient $[\![\partial d]\!]$ coincides with the gradient of the
interpretation $\partial [\![d]\!]$ for all sums of diagrams $d \in
\mathbf{C}_\Sigma^+$. We depict the gradient $\partial d$ as a
bubble surrounding the diagram $d$, we introduce bubbles formally in
section~\ref{4-bubbles}.
Once translated to string diagrams, the axioms for derivations on monoidal
categories with sums become:
$$\input{2-1-product-rule.tikz}$$

\section{Differentiating ZX}\label{2b-differentiating-zx}

This section applies the dual number construction to the diagrams of the ZX-calculus.

\begin{definition}
The diagrams of the ZX-calculus with smooth maps $\R^n \to \R$ as phases
form a category $\mathbf{ZX}_n = \mathbf{C}_\Sigma$ where
$\Sigma = \{ H : x \to x, \s \sigma : x^{\otimes 2} \to x^{\otimes 2} \}
+ \{ Z^{m, n}(\alpha) : x^{\otimes m} \to x^{\otimes n}
\ \vert \ m, n \in \N, \alpha : \R^n \to \R \}$.
$H$ is depicted as a yellow square, $\sigma$ as a swap and $Z^{m, n}(\alpha)$
as a green spider.
The interpretation $[\![-]\!]  : \mathbf{ZX}_n \to \mathbf{Mat}_\S$
in matrices over $\S = \R^n \to \C$ is given by on objects by $[\![x]\!] = 2$
and on arrows by $[\![H]\!] = \frac{1}{\sqrt{2}} \big(
\ket{0}\bra{0} + \ket{0}\bra{1} + \ket{1}\bra{0} - \ket{1}\bra{1}\big)$,
$[\![\sigma]\!] = \sum_{i,j \in \{ 0, 1 \}} \ket{j, i}\bra{i, j}$
and $[\![Z^{m, n}(\alpha)]\!] =
e^{-i \alpha / 2} \ket{0}^{\otimes n} \bra{0}^{\otimes m}
+ e^{i \alpha / 2} \ket{1}^{\otimes n} \bra{1}^{\otimes m}$.
We write $\mathbf{ZX}_n^+$ for the category of formal sums of parametrised
ZX diagrams.
\end{definition}

\begin{remark}
Note that we've scaled the standard interpretation of the green spider by a global phase
to match the usual definition of rotation gates in quantum circuits.
\end{remark}

\begin{remark}
For $n = 0$ we get $\mathbf{ZX}_0 = \mathbf{ZX}$ the ZX-calculus
with no parameters.
By currying, any ZX diagram $d \in \mathbf{ZX}_n$ can be seen as a function
$d : \R^n \to \text{Ar}(\mathbf{ZX})$ such that
$[\![-]\!] \circ d : \R^n \to \mathbf{Mat}_\C$ is smooth.
\end{remark}

\begin{lemma}\label{lemma-scalars}
A function $s : \R^n \to \C$ can be drawn as a scalar diagram in
$\mathbf{ZX}_n$ if and only if it is bounded.
\end{lemma}

\begin{proof}
Generalising \cite[P.~8.101]{CoeckeKissinger17} to parametrised scalars,
if there is a $k \in \N$ with $\vert s(\theta) \vert \leq 2^k$ for all
$\theta \in \R^n$ then there are parametrised
phases $\alpha, \beta : \R^n \to \R$ such that

\ctikzfig{2-2-bounded-lemma}

In the other direction, take any scalar diagram $d$ in $\mathbf{ZX}_n$.
Let $k$ be the number of spider in the diagram and $l$ the maximum number
of legs. By decomposing each spider as a sum of two disconnected diagrams,
we can write $d$ as a sum of $2^k$ diagrams. Each term of the sum is a product
of at most $\frac{1}{2} \times k \times l$ bone-shaped scalars. Each bone is
bounded by $2$, thus $[\![d]\!] : \R^n \to \C$ is bounded by $2^{k \times l}$.
\end{proof}

\begin{lemma}\label{lemma-rotations}
In $\mathbf{ZX}_n$, we have
$ \begin{tikzpicture}[scale=0.4, baseline={([yshift=-0.5ex]current bounding box.center)}]
	\begin{pgfonlayer}{nodelayer}
		\node [style={Z_med}] (11) at (1.775, 0) {$\alpha$};
		\node [style=none] (14) at (1.025, 0) {};
		\node [style=none] (15) at (1.775, 0.75) {};
		\node [style=none] (16) at (2.525, 0) {};
		\node [style=none] (17) at (1.775, -0.75) {};
		\node [style=none] (18) at (3.525, 0) {};
	\end{pgfonlayer}
	\begin{pgfonlayer}{edgelayer}
		\draw [style=boxedge] (16.center)
			 to [bend left=45, looseness=1.25] (17.center)
			 to [bend left=45, looseness=1.25] (14.center)
			 to [bend left=45, looseness=1.25] (15.center)
			 to [bend left=45, looseness=1.25] cycle;
		\draw (18.center) to (11);
	\end{pgfonlayer}
\end{tikzpicture} = \begin{tikzpicture}[scale=0.4, baseline={([yshift=-0.5ex]current bounding box.center)}]
	\begin{pgfonlayer}{nodelayer}
		\node [style={Z_long}] (11) at (1.775, 0) {$\alpha+\pi$};
		\node [style=none] (16) at (4.025, 0) {};
		\node [style=scalar, font={\tiny}] (20) at (-0.475, 0) {$\frac{\partial\alpha}{2}$};
	\end{pgfonlayer}
	\begin{pgfonlayer}{edgelayer}
		\draw (11) to (16.center);
	\end{pgfonlayer}
\end{tikzpicture} $
for all affine phases $\alpha : \R^n \to \R$.
\end{lemma}

\begin{proof}
$\partial [\![ Z(\alpha) ]\!]
= \partial \big( e^{-i \alpha / 2} \ket{0}+ e^{i \alpha / 2} \ket{1}\big)
= \frac{i\partial\alpha}{2}\big(-e^{-i \alpha / 2} \ket{0}
+ e^{i\alpha / 2} \ket{1}\big)
= \frac{\partial\alpha}{2}\big(e^{-i\frac{\alpha+\pi}{2}} \ket{0}
+ e^{i\frac{\alpha+\pi}{2}} \ket{1}\big)$.\\
$\alpha$ is affine so $\partial \alpha$ is constant, hence
bounded and from lemma~\ref{lemma-scalars} we know it can be drawn
in $\mathbf{ZX}_n$.
\end{proof}

\begin{theorem}\label{theorem-zx-diag-diff}
The ZX-calculus with affine maps $\R^n \to \R$ as phases admits diagrammatic
differentiation.
\end{theorem}

\begin{proof}
The Hadamard $H$ and swap $\sigma$ have derivative zero.
For the green spiders, we can extend lemma~\ref{lemma-rotations} from
single qubit rotations to arbitrary many legs using spider fusion:
$$\begin{tikzpicture}[scale=0.4, baseline={([yshift=-0.5ex]current bounding box.center)}]
	\begin{pgfonlayer}{nodelayer}
		\node [style=none] (14) at (0.75, 0) {};
		\node [style=none] (15) at (1.5, 0.75) {};
		\node [style=none] (16) at (2.25, 0) {};
		\node [style=none] (17) at (1.5, -0.75) {};
		\node [style=none] (0) at (0, 1.25) {};
		\node [style=none] (1) at (0, 0.5) {};
		\node [style=none] (2) at (0, -1.25) {};
		\node [style=none] (8) at (3, 1.25) {};
		\node [style=none] (9) at (3, 0.5) {};
		\node [style=none] (10) at (3, -1.25) {};
		\node [style={Z_med}] (11) at (1.5, 0) {$\alpha$};
		\node [style=node] (12) at (0.25, -0.25) {$\vdots$};
		\node [style=node] (13) at (2.75, -0.25) {$\vdots$};
	\end{pgfonlayer}
	\begin{pgfonlayer}{edgelayer}
		\draw [style=boxedge] (16.center)
			 to [bend left=45, looseness=1.25] (17.center)
			 to [bend left=45, looseness=1.25] (14.center)
			 to [bend left=45, looseness=1.25] (15.center)
			 to [bend left=45, looseness=1.25] cycle;
		\draw [bend left] (0.center) to (11);
		\draw [bend right=15] (11) to (1.center);
		\draw [bend left] (11) to (2.center);
		\draw [bend right] (11) to (10.center);
		\draw [bend left=15] (11) to (9.center);
		\draw [bend left] (11) to (8.center);
	\end{pgfonlayer}
\end{tikzpicture}
= \begin{tikzpicture}[scale=0.4, baseline={([yshift=-2.2ex]current bounding box.center)}]
	\begin{pgfonlayer}{nodelayer}
		\node [style=none] (15) at (0.75, 2) {};
		\node [style=none] (16) at (1.5, 2.75) {};
		\node [style=none] (17) at (2.25, 2) {};
		\node [style=none] (18) at (1.5, 1.25) {};
		\node [style=none] (0) at (0, 1.25) {};
		\node [style=none] (1) at (0, 0.5) {};
		\node [style=none] (2) at (0, -1.25) {};
		\node [style=none] (8) at (3, 1.25) {};
		\node [style=none] (9) at (3, 0.5) {};
		\node [style=none] (10) at (3, -1.25) {};
		\node [style={Z_med}] (11) at (1.5, 0) {};
		\node [style=node] (12) at (0.25, -0.25) {$\vdots$};
		\node [style=node] (13) at (2.75, -0.25) {$\vdots$};
		\node [style={Z_med}] (14) at (1.5, 2) {$\alpha$};
	\end{pgfonlayer}
	\begin{pgfonlayer}{edgelayer}
		\draw [style=boxedge] (17.center)
			 to [bend left=45, looseness=1.25] (18.center)
			 to [bend left=45, looseness=1.25] (15.center)
			 to [bend left=45, looseness=1.25] (16.center)
			 to [bend left=45, looseness=1.25] cycle;
		\draw [bend left] (0.center) to (11);
		\draw [bend right=15] (11) to (1.center);
		\draw [bend left] (11) to (2.center);
		\draw [bend right] (11) to (10.center);
		\draw [bend left=15] (11) to (9.center);
		\draw [bend left] (11) to (8.center);
		\draw (14) to (11);
	\end{pgfonlayer}
\end{tikzpicture}
= \begin{tikzpicture}[scale=0.4, baseline={([yshift=-2ex]current bounding box.center)}]
	\begin{pgfonlayer}{nodelayer}
		\node [style=none] (0) at (0, 1.25) {};
		\node [style=none] (1) at (0, 0.5) {};
		\node [style=none] (2) at (0, -1.25) {};
		\node [style=none] (8) at (3, 1.25) {};
		\node [style=none] (9) at (3, 0.5) {};
		\node [style=none] (10) at (3, -1.25) {};
		\node [style={Z_med}] (11) at (1.5, 0) {};
		\node [style=node] (12) at (0.25, -0.25) {$\vdots$};
		\node [style=node] (13) at (2.75, -0.25) {$\vdots$};
		\node [style={Z_long}] (14) at (1.5, 2) {$\alpha + \pi$};
		\node [style=none] (15) at (-1.25, 1.25) {};
		\node [style=none] (16) at (-1.25, 0.5) {};
		\node [style=none] (17) at (-1.25, -1.25) {};
		\node [style=scalar, font={\tiny}] (19) at (3.5, 2.25) {$\frac{\partial\alpha}{2}$};
		\node [style=none] (20) at (4.25, 1.25) {};
		\node [style=none] (21) at (4.25, 0.5) {};
		\node [style=none] (22) at (4.25, -1.25) {};
	\end{pgfonlayer}
	\begin{pgfonlayer}{edgelayer}
		\draw (15.center)
			 to (0.center)
			 to [bend left] (11);
		\draw (11)
			 to [bend right=15] (1.center)
			 to (16.center);
		\draw (17.center)
			 to (2.center)
			 to [bend right] (11);
		\draw (11)
			 to [bend right] (10.center)
			 to (22.center);
		\draw (21.center)
			 to (9.center)
			 to [bend right=15] (11);
		\draw (11)
			 to [bend left] (8.center)
			 to (20.center);
		\draw (14) to (11);
	\end{pgfonlayer}
\end{tikzpicture}
= \begin{tikzpicture}[scale=0.4, baseline={([yshift=-1.9ex]current bounding box.center)}]
	\begin{pgfonlayer}{nodelayer}
		\node [style=none] (0) at (-0.5, 1.25) {};
		\node [style=none] (1) at (-0.5, 0.5) {};
		\node [style=none] (2) at (-0.5, -1.25) {};
		\node [style=none] (8) at (3.5, 1.25) {};
		\node [style=none] (9) at (3.5, 0.5) {};
		\node [style=none] (10) at (3.5, -1.25) {};
		\node [style={Z_long}] (11) at (1.5, 0) {$\alpha+\pi$};
		\node [style=node] (12) at (-0.25, -0.25) {$\vdots$};
		\node [style=node] (13) at (3.25, -0.25) {$\vdots$};
		\node [style=none] (14) at (1, 0) {};
		\node [style=none] (15) at (2, 0) {};
		\node [style=scalar, font={\tiny}] (16) at (1.5, 1.75) {$\frac{\partial\alpha}{2}$};
	\end{pgfonlayer}
	\begin{pgfonlayer}{edgelayer}
		\draw [bend left] (0.center) to (11);
		\draw [bend left] (11) to (2.center);
		\draw [bend right] (11) to (10.center);
		\draw [bend left] (11) to (8.center);
		\draw [bend left] (1.center) to (14.center);
		\draw [bend left] (15.center) to (9.center);
	\end{pgfonlayer}
\end{tikzpicture}$$
\end{proof}

Note that there is no diagrammatic differentiation for the ZX-calculus with
smooth maps as phases, even when restricted to bounded functions.
Take for example $\alpha : \R \to \R$ with $\alpha(\theta) = \sin \theta^2$,
it is smooth and bounded by $1$ but its derivative $\partial \alpha$ is
unbounded.
Thus, from lemma~\ref{lemma-scalars} we know it cannot be represented as a
scalar diagram in $\mathbf{ZX}_1$: there can be no diagrammatic
differentiation $\partial : \mathbf{ZX}_1 \to \D[\mathbf{ZX}_1]$.
In such cases, we can always extend the signature by adjoining a new box
for each derivative.

\begin{proposition}
For every interpretation $[\![-]\!] : \mathbf{C}_\Sigma^+ \to \mathbf{Mat}_\S$,
there is an extended signature $\Sigma' \supset \Sigma$
and interpretation $[\![-]\!] : \mathbf{C}_{\Sigma'}^+ \to \mathbf{Mat}_\S$
such that $\mathbf{C}_{\Sigma'}^+$ admits digrammatic differentiation.
\end{proposition}

\begin{proof}
Let $\Sigma' = \cup_{n \in \N} \Sigma^n$ where $\Sigma^0 = \Sigma$
and $\Sigma^{n + 1} = \Sigma^n \cup \{ \partial f \ \vert \ f \in \Sigma^n \}$
with $[\![\partial f]\!] = \partial [\![f]\!]$.
\end{proof}

The issue of being able to represent arbitrary scalars disappears if we work
with the algebraic ZX-calculus instead. Furthermore, we can generalise
from $\S = \R^n \to \C$ to any commutative rig.

\begin{definition}
The diagrams of the algebraic ZX-calculus over a commutative rig $\S$ form a
category $\mathbf{ZX}_\S = \mathbf{C}_\Sigma$ where the signature $\Sigma$ is
given in \cite[Table 2]{Wang20} and the interpretation
is given in \cite[§6]{Wang20}.
In particular, there is a green square $R_Z^{m, n}(a) \in \Sigma_1$ for each $a \in S$
and $m, n \in \N$ with $[\![R_Z^{m, n}(a)]\!] =
\ket{0}^{\otimes n} \bra{0}^{\otimes m}
+ a \ket{1}^{\otimes n} \bra{1}^{\otimes m}$.
Let $\mathbf{ZX}_\S^+$ be the category of formal sums of algebraic ZX
diagrams over $\S$.
\end{definition}

\begin{theorem}
Diagrammatic derivations on $[\![-]\!] : \mathbf{ZX}_\S^+ \to \mathbf{Mat}_\S$
are in one-to-one correspondance with rig derivations $\partial : \S \to \S$.
\end{theorem}

\begin{proof}
Given a derivation $\partial$ on $\S$, we have
$\partial [\![R_Z^{m, n}(a)]\!]
= (\partial a) \ket{1}^{\otimes n} \bra{1}^{\otimes m}$
and $\partial a$ can be represented by the scalar diagram
$R_Z^{1, 0}(\partial a) \ket{1}$.
In the other direction, a diagrammatic derivation $\partial$ on
$\mathbf{ZX}_\S^+$ is uniquely determined by its action on scalars
$R_Z^{1, 0}(a) \ket{1}$ for $a \in \S$.
\end{proof}

One application of diagrammatic differentiation is to solve
differential equations between diagrams. As a first step,
we apply Stone's theorem \cite{Stone32} on one-parameter unitary groups
to the ZX-calculus.

\begin{definition}
A one-parameter unitary group is a unitary matrix $U : n \to n$
in $\mathbf{Mat}_{\R \to \C}$ with $U(0) = \id_n$ and $U(\theta) U(\theta') = U(\theta + \theta')$
for all $\theta, \theta' \in \R$. It is strongly continuous when
$\lim_{\theta \to \theta_0} U(\theta) = U(\theta_0)$ for all $\theta_0 \in \R$.
We say a one-parameter diagram $d : x^{\otimes n} \to x^{\otimes n}$
is a unitary group if its interpretation $[\![d]\!]$ is.
\end{definition}

\begin{remark}
The interpretation of diagrams with smooth maps as phases must be strongly continuous.
\end{remark}

\begin{theorem}[Stone]
There is a one-to-one correspondance between strongly continuous one-parameter
unitary groups $U : n \to n$ in $\mathbf{Mat}_{\R \to \C}$ and self-adjoint
matrices $H : n \to n$ in $\mathbf{Mat}_{\C}$. The bijection is given
explicitly by $U(\theta) = \exp(i \theta H)$ and $H = - i (\partial U)(0)$,
translated in terms of diagrams with bubbles we get:
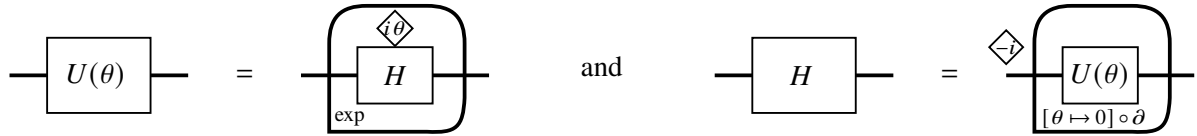
\ctikzfig{2-5-stone-theorem}
\end{theorem}

\begin{corollary}
A one-parameter diagram $d : x^{\otimes n} \to x^{\otimes n}$ in
$\mathbf{ZX}_1$ is a unitary group if and only if there is a constant
self-adjoint diagram
$h : x^{\otimes n} \to x^{\otimes n}$ such that $\partial d = i h \then d$.
\end{corollary}

\begin{proof}
Given the diagram for a unitary group $d$, we compute its diagrammatic
differentiation $\partial d$ and get $h$ by pattern matching.
Conversely given a self-adjoint $h$, the diagram $d = \exp(i \theta h)$
is a unitary group.
\end{proof}

\begin{example}
Let $d = R_z(\alpha) \otimes R_x(\alpha)$ for a smooth $\alpha : \R \to \R$,
then the following implies $d(\theta) = \exp(i \theta h)$\linebreak
$
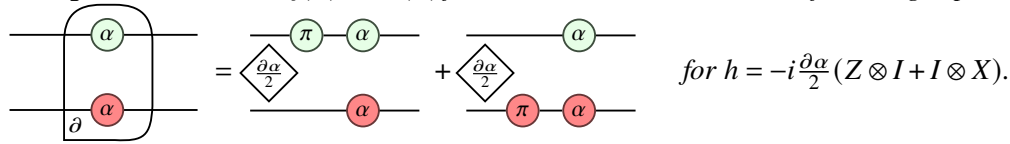
\begin{tikzpicture}[scale=1, baseline={([yshift=-0.5ex]current bounding box.center)}]
	\begin{pgfonlayer}{nodelayer}
		\node [style=none] (46) at (4.225, 0.4) {};
		\node [style=none] (47) at (4.4, 0.4) {};
		\node [style=none] (48) at (4.4, -1.4) {};
		\node [style=none] (49) at (4.9, -0.1) {};
		\node [style=none] (50) at (4.9, -0.95) {};
		\node [style=none] (51) at (3.725, -0.1) {};
		\node [style=none] (52) at (3.725, -1.4) {};
		\node [style=none] (17) at (3, 0) {};
		\node [style=none] (19) at (3, -1) {};
		\node [style=none] (36) at (5.5, -1) {};
		\node [style=none] (37) at (5.5, 0) {};
		\node [style={Z_med}] (44) at (4.3, 0) {$\alpha$};
		\node [style={X_med}] (45) at (4.3, -1) {$\alpha$};
		\node [style=none, font={\scriptsize}] (56) at (3.875, -1.2) {$\partial$};
	\end{pgfonlayer}
	\begin{pgfonlayer}{edgelayer}
		\draw [style=boxedge] (50.center)
			 to (49.center)
			 to [bend right=45, looseness=1.25] (47.center)
			 to (46.center)
			 to [bend right=45, looseness=1.25] (51.center)
			 to (52.center)
			 to (48.center)
			 to [bend right=45, looseness=1.50] cycle;
		\draw (19.center) to (36.center);
		\draw (17.center) to (37.center);
	\end{pgfonlayer}
\end{tikzpicture}
= \begin{tikzpicture}[scale=1, baseline={([yshift=-0.5ex]current bounding box.center)}]
	\begin{pgfonlayer}{nodelayer}
		\node [style=none] (17) at (3, 0) {};
		\node [style=none] (19) at (3, -1) {};
		\node [style=none] (36) at (5.25, -1) {};
		\node [style=none] (37) at (5.25, 0) {};
		\node [style={Z_med}] (44) at (4.5, 0) {$\alpha$};
		\node [style={X_med}] (45) at (4.5, -1) {$\alpha$};
		\node [style={Z_med}] (46) at (3.75, 0) {$\pi$};
		\node [style=scalar] (47) at (3.25, -0.5) {$\frac{\partial\alpha}{2}$};
	\end{pgfonlayer}
	\begin{pgfonlayer}{edgelayer}
		\draw (19.center) to (36.center);
		\draw (17.center) to (46);
		\draw (46) to (44);
		\draw (44) to (37.center);
	\end{pgfonlayer}
\end{tikzpicture}
+ \begin{tikzpicture}[scale=1, baseline={([yshift=-0.5ex]current bounding box.center)}]
	\begin{pgfonlayer}{nodelayer}
		\node [style=none] (17) at (3, -1) {};
		\node [style=none] (19) at (3, 0) {};
		\node [style=none] (36) at (5.25, 0) {};
		\node [style=none] (37) at (5.25, -1) {};
		\node [style={X_med}] (44) at (4.5, -1) {$\alpha$};
		\node [style={Z_med}] (45) at (4.5, 0) {$\alpha$};
		\node [style={X_med}] (46) at (3.75, -1) {$\pi$};
		\node [style=scalar] (47) at (3.25, -0.5) {$\frac{\partial\alpha}{2}$};
	\end{pgfonlayer}
	\begin{pgfonlayer}{edgelayer}
		\draw (19.center) to (36.center);
		\draw (17.center) to (46);
		\draw (46) to (44);
		\draw (44) to (37.center);
	\end{pgfonlayer}
\end{tikzpicture} \quad$
for $h = - i \frac{\partial \alpha}{2}(Z \otimes I + I \otimes X)$.
\end{example}

\begin{example}
Let $d = P(\alpha, ZX)$ be a Pauli gadget as defined in \cite[def.~4.1]{CowtanEtAl20a} then
the following implies
$d(\theta) = \exp(i \theta h)$ for $h = -i \frac{\partial \alpha}{2} Z \otimes X$.
$\begin{tikzpicture}[scale=1, baseline={([yshift=-0.5ex]current bounding box.center)}]
	\begin{pgfonlayer}{nodelayer}
		\node [style=none] (10) at (0.975, 1.4) {};
		\node [style=none] (11) at (1.9, 1.4) {};
		\node [style=none] (12) at (1.9, -0.4) {};
		\node [style=none] (13) at (2.4, 0.9) {};
		\node [style=none] (14) at (2.4, 0.05) {};
		\node [style=none] (15) at (0.475, 0.9) {};
		\node [style=none] (16) at (0.475, -0.4) {};
		\node [style=none] (0) at (0, 1) {};
		\node [style=none] (1) at (0, 0) {};
		\node [style=none] (2) at (2.75, 0) {};
		\node [style=none] (3) at (2.75, 1) {};
		\node [style=Z] (4) at (1, 1) {};
		\node [style=X] (5) at (1, 0) {};
		\node [style=X] (6) at (1.5, 0.5) {};
		\node [style={Z_med}] (7) at (2, 0.5) {};
		\node [style={Z_med}] (8) at (2, 0.5) {$\alpha$};
		\node [style=H, rotate=45] (9) at (1.25, 0.25) {};
		\node [style=none, font={\scriptsize}] (17) at (0.625, -0.2) {$\partial$};
	\end{pgfonlayer}
	\begin{pgfonlayer}{edgelayer}
		\draw [style=boxedge] (14.center)
			 to (13.center)
			 to [bend right=45, looseness=1.25] (11.center)
			 to (10.center)
			 to [bend right=45, looseness=1.25] (15.center)
			 to (16.center)
			 to (12.center)
			 to [bend right=45, looseness=1.50] cycle;
		\draw (4) to (6);
		\draw (5) to (6);
		\draw (6) to (8);
		\draw (0.center) to (4);
		\draw (1.center) to (2.center);
		\draw (3.center) to (4);
	\end{pgfonlayer}
\end{tikzpicture}
= \begin{tikzpicture}[scale=1, baseline={([yshift=-0.5ex]current bounding box.center)}]
	\begin{pgfonlayer}{nodelayer}
		\node [style=none] (0) at (0, 1) {};
		\node [style=none] (1) at (0, 0) {};
		\node [style=none] (2) at (2.75, 0) {};
		\node [style=none] (3) at (2.75, 1) {};
		\node [style=Z] (4) at (0.75, 1) {};
		\node [style=X] (5) at (0.75, 0) {};
		\node [style=X] (6) at (1.25, 0.5) {};
		\node [style={Z_med}] (7) at (2, 0.5) {};
		\node [style={Z_long}] (8) at (2, 0.5) {$\alpha+\pi$};
		\node [style=H, rotate=45] (9) at (1, 0.25) {};
		\node [style=scalar] (10) at (0.25, 0.5) {$\frac{\partial\alpha}{2}$};
	\end{pgfonlayer}
	\begin{pgfonlayer}{edgelayer}
		\draw (4) to (6);
		\draw (5) to (6);
		\draw (6) to (8);
		\draw (0.center) to (4);
		\draw (1.center) to (2.center);
		\draw (3.center) to (4);
	\end{pgfonlayer}
\end{tikzpicture}
= \begin{tikzpicture}[scale=1, baseline={([yshift=-0.5ex]current bounding box.center)}]
	\begin{pgfonlayer}{nodelayer}
		\node [style=none] (0) at (0, 1) {};
		\node [style=none] (1) at (0, 0) {};
		\node [style=none] (2) at (2.5, 0) {};
		\node [style=none] (3) at (2.5, 1) {};
		\node [style=Z] (4) at (1, 1) {};
		\node [style=X] (5) at (1, 0) {};
		\node [style=X] (6) at (1.5, 0.5) {};
		\node [style={Z_med}] (8) at (2, 0.5) {$\alpha$};
		\node [style=H, rotate=45] (9) at (1.25, 0.25) {};
		\node [style={Z_med}] (10) at (0.5, 1) {$\pi$};
		\node [style={X_med}] (11) at (0.5, 0) {$\pi$};
		\node [style=scalar] (12) at (0, 0.5) {$\frac{\partial\alpha}{2}$};
	\end{pgfonlayer}
	\begin{pgfonlayer}{edgelayer}
		\draw (4) to (6);
		\draw (5) to (6);
		\draw (6) to (8);
		\draw (0.center) to (4);
		\draw (1.center) to (2.center);
		\draw (3.center) to (4);
	\end{pgfonlayer}
\end{tikzpicture}$
\end{example}


\section{Differentiating quantum circuits}\label{3-dual-circuits}

In this section, we extend diagrammatic differentiation to classical-quantum
circuits. These circuit diagrams have two kinds of wires for bits and qubits,
and boxes for pure quantum processes, measurements and preparations.
We interpret these classical-quantum circuits in terms of parametrised matrices,
where the tensor product reorders the indices to keep the classical and quantum
dimensions in order. Borrowing the term from Coecke and Kissinger
\cite{CoeckeKissinger17}, we call these matrices cq-maps.
In this context, diagrammatic derivations correspond to the notion of gradient
recipe for parametrised quantum gates \cite{SchuldEtAl19}.

We first give the definition of parametrised cq-maps which is at the basis
of our Python implementation.
The category $\mathbf{CQMap}_n$ has objects given by pairs of natural numbers
$\text{Ob}(\mathbf{CQMap}_n) = \N \times \N$, where the first and second element
of the pair encode the classical and the quantum dimension of the system respectively.
Arrows $f : (a, b) \to (c, d)$ are given by $a \times b^2 \to c \times d^2$
parametrised complex matrices, i.e. with entries in $\R^n \to \C$.
Composition of cq-maps is given by multiplying their underlying matrices.
Tensor is given on objects by pointwise multiplication and on arrows by the
following diagram in $\Mat_{\R^n \to \C}$:
$$\begin{tikzpicture}[scale=0.8]
	\begin{pgfonlayer}{nodelayer}
		\node [style=none] (0) at (-3, 1.5) {};
		\node [style=none] (1) at (-1, 1.5) {};
		\node [style=none] (2) at (-1, -1.5) {};
		\node [style=none] (3) at (-3, -1.5) {};
		\node [style=none] (4) at (-3, 1) {};
		\node [style=none] (5) at (-3, 0) {};
		\node [style=none] (6) at (-3, -1) {};
		\node [style=none] (7) at (-1, -1) {};
		\node [style=none] (8) at (-1, 0) {};
		\node [style=none] (9) at (-1, 1) {};
		\node [style=none] (10) at (0, 1) {};
		\node [style=none] (11) at (-4, 1) {};
		\node [style=none] (12) at (-4, 0) {};
		\node [style=none] (13) at (0, 0) {};
		\node [style=none] (14) at (0, -1) {};
		\node [style=none] (15) at (-4, -1) {};
		\node [style=none] (16) at (-2, 0) {$f \otimes f'$};
		\node [style=none] (17) at (-5, 1) {$a \times a'$};
		\node [style=none] (18) at (-5, 0) {$b \times b'$};
		\node [style=none] (19) at (-5, -1) {$b \times b'$};
		\node [style=none] (20) at (1, 1) {$c \times c'$};
		\node [style=none] (21) at (1, 0) {$d \times d'$};
		\node [style=none] (22) at (1, -1) {$d \times d'$};
		\node [style=none] (23) at (3, 0) {$=$};
		\node [style=none] (24) at (8.25, 1.75) {};
		\node [style=none] (25) at (9.75, 1.75) {};
		\node [style=none] (26) at (9.75, 0.25) {};
		\node [style=none] (27) at (8.25, 0.25) {};
		\node [style=none] (28) at (8.25, 1.5) {};
		\node [style=none] (29) at (8.25, 1) {};
		\node [style=none] (30) at (8.25, 0.5) {};
		\node [style=none] (31) at (9.75, 0.5) {};
		\node [style=none] (32) at (9.75, 1) {};
		\node [style=none] (33) at (9.75, 1.5) {};
		\node [style=none] (35) at (7.5, 1.5) {};
		\node [style=none] (36) at (7.5, 1) {};
		\node [style=none] (39) at (7.5, 0.5) {};
		\node [style=none] (40) at (9, 1) {$f$};
		\node [style=none] (41) at (7.75, 1.75) {$a$};
		\node [style=none] (42) at (7.75, 1.25) {$b$};
		\node [style=none] (43) at (7.75, 0.75) {$b$};
		\node [style=none] (44) at (10.25, 1.75) {$c$};
		\node [style=none] (45) at (10.25, 1.25) {$d$};
		\node [style=none] (46) at (10.25, 0.75) {$d$};
		\node [style=none] (70) at (5, 1.5) {$a$};
		\node [style=none] (71) at (5, 1) {$a'$};
		\node [style=none] (72) at (5, 0.25) {$b$};
		\node [style=none] (73) at (5, -0.25) {$b'$};
		\node [style=none] (74) at (5, -1.5) {$b'$};
		\node [style=none] (75) at (5, -1) {$b$};
		\node [style=none] (82) at (8.25, -0.25) {};
		\node [style=none] (83) at (9.75, -0.25) {};
		\node [style=none] (84) at (9.75, -1.75) {};
		\node [style=none] (85) at (8.25, -1.75) {};
		\node [style=none] (86) at (8.25, -0.5) {};
		\node [style=none] (87) at (8.25, -1) {};
		\node [style=none] (88) at (8.25, -1.5) {};
		\node [style=none] (89) at (9.75, -1.5) {};
		\node [style=none] (90) at (9.75, -1) {};
		\node [style=none] (91) at (9.75, -0.5) {};
		\node [style=none] (93) at (7.5, -0.5) {};
		\node [style=none] (94) at (7.5, -1) {};
		\node [style=none] (97) at (7.5, -1.5) {};
		\node [style=none] (98) at (9, -1) {$f'$};
		\node [style=none] (99) at (7.75, -0.25) {$a'$};
		\node [style=none] (100) at (7.75, -0.75) {$b'$};
		\node [style=none] (101) at (7.75, -1.25) {$b'$};
		\node [style=none] (102) at (10.25, -0.25) {$c'$};
		\node [style=none] (103) at (10.25, -0.75) {$d'$};
		\node [style=none] (104) at (10.25, -1.25) {$d'$};
		\node [style=none] (105) at (5.5, 1.5) {};
		\node [style=none] (106) at (5.5, 1) {};
		\node [style=none] (107) at (5.5, 0.25) {};
		\node [style=none] (108) at (5.5, -0.25) {};
		\node [style=none] (109) at (5.5, -1) {};
		\node [style=none] (110) at (5.5, -1.5) {};
		\node [style=none] (111) at (10.5, 1.5) {};
		\node [style=none] (112) at (10.5, 1) {};
		\node [style=none] (113) at (10.5, 0.5) {};
		\node [style=none] (123) at (10.5, -0.5) {};
		\node [style=none] (124) at (10.5, -1) {};
		\node [style=none] (125) at (10.5, -1.5) {};
		\node [style=none] (129) at (12.5, 1.5) {};
		\node [style=none] (130) at (12.5, 1) {};
		\node [style=none] (131) at (12.5, 0.25) {};
		\node [style=none] (132) at (12.5, -0.25) {};
		\node [style=none] (133) at (12.5, -1) {};
		\node [style=none] (134) at (12.5, -1.5) {};
		\node [style=none] (135) at (13, 1.5) {$c$};
		\node [style=none] (136) at (13, 0.25) {$d$};
		\node [style=none] (137) at (13, -1) {$d$};
		\node [style=none] (138) at (13, 1) {$c'$};
		\node [style=none] (139) at (13, -0.25) {$d'$};
		\node [style=none] (140) at (13, -1.5) {$d'$};
	\end{pgfonlayer}
	\begin{pgfonlayer}{edgelayer}
		\draw (0.center) to (3.center);
		\draw (1.center) to (2.center);
		\draw (3.center) to (2.center);
		\draw (0.center) to (1.center);
		\draw (11.center) to (4.center);
		\draw (12.center) to (5.center);
		\draw (15.center) to (6.center);
		\draw (7.center) to (14.center);
		\draw (13.center) to (8.center);
		\draw (10.center) to (9.center);
		\draw (24.center) to (27.center);
		\draw (25.center) to (26.center);
		\draw (27.center) to (26.center);
		\draw (24.center) to (25.center);
		\draw (35.center) to (28.center);
		\draw (36.center) to (29.center);
		\draw (39.center) to (30.center);
		\draw (82.center) to (85.center);
		\draw (83.center) to (84.center);
		\draw (85.center) to (84.center);
		\draw (82.center) to (83.center);
		\draw (93.center) to (86.center);
		\draw (94.center) to (87.center);
		\draw (97.center) to (88.center);
		\draw [in=180, out=0] (105.center) to (35.center);
		\draw [in=180, out=0, looseness=0.75] (106.center) to (93.center);
		\draw [in=-180, out=0] (107.center) to (36.center);
		\draw [in=180, out=0] (108.center) to (94.center);
		\draw [in=-180, out=0, looseness=0.75] (109.center) to (39.center);
		\draw [in=-180, out=0] (110.center) to (97.center);
		\draw [in=0, out=180] (129.center) to (111.center);
		\draw [in=0, out=180, looseness=0.75] (130.center) to (123.center);
		\draw [in=0, out=180] (131.center) to (112.center);
		\draw [in=0, out=180] (132.center) to (124.center);
		\draw [in=0, out=180, looseness=0.75] (133.center) to (113.center);
		\draw [in=0, out=180] (134.center) to (125.center);
		\draw (33.center) to (111.center);
		\draw (32.center) to (112.center);
		\draw (31.center) to (113.center);
		\draw (91.center) to (123.center);
		\draw (90.center) to (124.center);
		\draw (89.center) to (125.center);
	\end{pgfonlayer}
\end{tikzpicture}$$

Each pure map $f : a \to b$ in $\Mat_{\R^n \to \C}$ embeds as a
cq-map $(1, a) \to (1, b)$ by ``doubling'', i.e. tensoring with its complex
conjugate $f \mapsto \bar{f} \otimes f$.
Note that doubling is faithful up to a global phase.
For each dimension $a \in \N$, there are distinguished cq-maps
$M_a : (1, a) \to (a, 1), \s E_a : (a, 1) \to (1, a)$ for measurement and
preparation in the computational basis with matrices given by
$M_a = \sum_{i < a} \ket{i} \bra{i, i}$
and $E_a = \sum_{i < a} \ket{i, i} \bra{i}$.
The sum of two cq-maps is given by entrywise addition of their underlying
matrix. Note that doubling does not preserve sums,
i.e. $\overline{(\sum_i f_i)} \otimes (\sum_i f_i)
\neq \sum_i (\overline{f_i} \otimes f_i)$.
In quantum mechanical terms, this corresponds to the distinction between
quantum superposition and probabilistic mixing.

\begin{remark}
The cq-maps we have defined here differ from \cite{CoeckeKissinger17} in two
minor ways. First, we take the algebraic conjugate rather than the diagrammatic
conjugate, i.e. we take $\overline{f \otimes g}
= \overline{f} \otimes \overline{g} \neq \overline{g} \otimes \overline{f}$.
This is just a choice of convention that makes numerical computation easier.
Second, our category $\mathbf{CQMap}$ contains matrices that have no physical
interpretation, e.g. we do not ask for complete positivity. This can be fixed
by considering the subcategory in the image of the interpretation functor
defined below.
\end{remark}

Take a monoidal signature $\Sigma$ with one object $\Sigma_0 = \{ q \}$
interpreted as a qubit, and boxes interpreted as pure quantum processes with
$n$ parameters. That is, we fix a parametrised interpretation functor $[\![-]\!]
: \mathbf{C}_\Sigma \to \mathbf{Mat}_{\R^n \to \C}$ with $[\![q]\!] = 2$.
This could be the signatures for parametrised or algebraic ZX
from the previous section, or any universal quantum gate set plus boxes
for scalars, bras and kets.
We define an extended signature $cq(\Sigma) \supset \Sigma$ with two objects
$cq(\Sigma)_0 = \{ c, q \}$ interpreted as bit and qubit respectively.
Boxes are given by $cq(\Sigma)_1 =
\{ \hat{f} : q^{\otimes a} \to q^{\otimes b} \ \vert \ f \in \Sigma_1 \}
+ \{ M : q \to c, \s E : c \to q \}$.
Let $\mathbf{C}_{cq(\Sigma)}$ be the free monoidal category it generates,
i.e. arrows are classical-quantum circuits.
Their interpretation is given by a monoidal functor
$[\![-]\!] : \mathbf{C}_{cq(\Sigma)} \to \mathbf{CQMap}_n$ with
$[\![c]\!] = (2, 1)$ and $[\![q]\!] = (1, 2)$ on objects.
On arrows we define $[\![M]\!] = M_2$, $[\![E]\!] = E_2$ and
$[\![\hat{f}]\!] = \overline{[\![f]\!]} \otimes [\![f]\!]$.
We write $cq(\mathbf{ZX}_n)$ for the category of classical-quantum circuits
with parametrised ZX diagrams as pure processes.

Let $\mathbf{C}_{cq(\Sigma)}^+$ be the free monoidal category with sums,
i.e. arrows are bags of circuits.
Again, we want to find a diagrammatic derivation
$\partial : \mathbf{C}_{cq(\Sigma)}^+ \to \D[\mathbf{C}_{cq(\Sigma)}^+]$
which commutes with the interpretation, i.e. such that
$[\![\partial \hat{f}]\!] = \partial [\![\hat{f}]\!] =
\partial \big( \overline{[\![f]\!]} \otimes [\![f]\!] \big)$
for all pure maps $f \in \Sigma_1$.
Note that a diagrammatic derivation for pure processes in
$\mathbf{C}_{\Sigma}^+$ does not in general lift to one for classical-quantum
circuits in $\mathbf{C}_{\Sigma}$. Indeed, using the product rule we get
$\partial \big( \overline{[\![f]\!]} \otimes [\![f]\!] \big)
\s = \s \partial \overline{[\![f]\!]} \otimes [\![f]\!]
\ + \ \overline{[\![f]\!]} \otimes \partial [\![f]\!]
\s \neq \s \overline{[\![\partial f]\!]} \otimes [\![\partial f]\!]$.

Hence we need equations, called gradient recipes, to rewrite the gradient of a
pure map $\partial [\![\hat{f}]\!]$ as the pure map of a gradient
$[\![\partial \hat{f}]\!]$.
In the special case of Hermitian operators with at most two unique eigenvalues,
gradient recipes are given by the parameter-shift rule. In the general case
where the parameter-shift rule does not apply, gradient recipes require the
introduction of an ancilla qubit.

\begin{theorem}[Schuld et al.]
For a one-parameter unitary group $f$ with
$[\![f(\theta)]\!] = \exp (i \theta H)$, if $H$ has at most two eigenvalues
$\pm r$, then there is a shift $s \in [0, 2 \pi)$ such that
$[\![r\big(f(\theta + s) - f(\theta - s)\big)]\!] = \partial [\![f(\theta)]\!]$.
\end{theorem}

\begin{proof}
The shift is given by $s = \frac{\pi}{4 r}$, see the Taylor expansion given in
\cite[Theorem 1]{SchuldEtAl19}.
\end{proof}

\begin{corollary}
Classical-quantum circuits $cq(\mathbf{ZX}_n)$ with parametrised ZX diagrams as
pure processes admit diagrammatic differentiation.
\end{corollary}

\begin{proof}
The $Z$ rotation has eigenvalues $\pm 1$, hence the spiders with two legs have
diagrammatic differentiation given by the parameter-shift rule:

\ctikzfig{3-2-param-shift}

As for theorem~\ref{theorem-zx-diag-diff}, this extends to
arbitrary-many legs using spider fusion.
\end{proof}

\begin{remark}
All scalars in $cq(\mathbf{ZX}_n)$ are non-negative real numbers. Thus in order
to encode the substraction of the parameter shift-rule diagrammatically, we
need either to consider formal sums with minus signs (a.k.a. enrichment in
Abelian groups) or simply to extend the signature with the $-1$ scalar.
\end{remark}

\begin{example}
The quantum enhanced feature spaces of \cite{HavlicekEtAl19} are parametrised
classical-quantum circuits.
The quantum classifier can be drawn as a diagram:

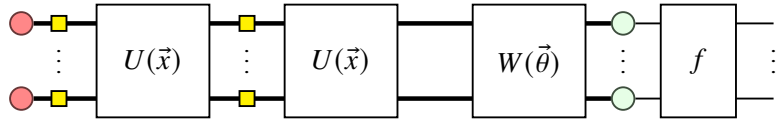
\ctikzfig{3-3-quantum-enhanced}

where $U(\vec{x})$ depends on the input, $W(\vec{\theta})$ depends on the
trainable parameters and $f$ is a fixed Boolean function encoded as a linear map.
\end{example}


\section{Bubbles and the chain rule} \label{4-bubbles}

This section introduces an extension of the language of string diagrams
that encodes arbitrary non-linear operators on matrices: bubbles.
Previous sections already used two kinds of bubbles informally:
matrix exponentials and gradients. We give a formal definition of bubbles and
their gradients with the chain rule. We then use them to compute the gradient
of hybrid classical-quantum circuits where the measurement results can be
post-processed by any classical feed-forward neural network.

Fix a set of colours $C$.
Take a monoidal signature $\Sigma$, we construct the free monoidal category with
sums and bubbles $\mathbf{C}_{\beta(\Sigma)}^+$, i.e. arrows are formal sums of
diagrams with bubbles. We define the signature of bubbled diagrams
as a union $\beta(\Sigma) = \bigcup_{n \in \N} \beta(\Sigma, n)$ where the
signature of $(\leq n)$-nested bubbles $\beta(\Sigma, n)$ is defined by
induction:
$$
\beta(\Sigma, 0) \s = \s \Sigma \qquad \text{and} \qquad
\beta(\Sigma, n + 1) \s = \s \big\{\beta^c(d) : x \to y \s \vert \s
c \in C, \s d : x \to y \in \mathbf{C}_{\beta(\Sigma, n)}^+ \big\}$$
That is, we put a formal sum of diagrams $d \in \mathbf{C}_{\beta(\Sigma, n)}^+$
with $(\leq n)$-nested bubbles inside a $c$-coloured bubble and take it as a box
$\beta^c(d) \in \beta(\Sigma, n + 1)$ for diagrams with $(n + 1)$-nested
bubbles.
We say a monoidal category $\mathbf{C}$ has bubbles when it comes equipped with
a unary operator on homsets $\beta^c :
\coprod_{x, y} \mathbf{C}(x, y) \to \mathbf{C}(x, y)$ for each colour $c \in C$.

Although it makes the bureaucracy heavier, we may consider bubbles that change
the domain and codomain of the diagram inside. Such a bubble is defined by two
operators on objects
$\beta^c_\dom, \beta^c_\cod : \text{Ob}(\mathbf{C}) \to \text{Ob}(\mathbf{C})$
and an operator on homsets $\beta^c : \coprod_{x, y}
\mathbf{C}(x, y) \to \mathbf{C}(\beta^c_\dom(x), \beta^c_\cod(y))$.

\begin{example}
Bubbles first appear in Penrose and Rindler \cite{PenroseRindler84}
where they are used to encode the covariant derivative. An extra wire comes in
the bubble to encode the dimension of the tangent vector.
\end{example}

\begin{example}
The functorial boxes of Melli\`es \cite{Mellies06} can be thought of as
well-behaved bubbles, i.e. such that the composition of bubbles is the
bubble of the composition. Indeed, a functor $F : \mathbf{C} \to \mathbf{D}$
between two categories $\mathbf{C}$ and $\mathbf{D}$ defines a bubble on the
subcategory of their coproduct $\mathbf{C} \coprod \mathbf{D}$ spanned by $\mathbf{C}$.
\end{example}

\begin{example}
Bubbles appear under the name ``uooh'' (unary operator on homsets) in
\cite{HaydonSobocinski20} where they are used to encode the sep lines of
C.S.Peirce's existential graphs.
Take the predicates of a first-order logic as signature, i.e. one
generating object $x$ and each predicate $P$ with arity $k$ as a box with
$\dom(P) = 1$ and $\cod(P) = x^{\otimes k}$. Add generators for spiders to
encode lines of identity.
Then bubbled diagrams encode first-order logic formulae, and every formula can
be represented in this way. Logical deduction rules may be given entirely
in terms of diagrammatic rules.
The evaluation of first-order logic formulae is a bubble-preserving functor
$F : \mathbf{C}_{B(\Sigma)} \to \Mat_\B$, where bubbles are interpreted
as pointwise negation.
\end{example}

\begin{example}
Take colours to be arbitrary rig-valued functions $\S \to \S$, then
the category of matrices $\mathbf{Mat}_\S$ has bubbles given by pointwise
application. Gradient bubbles $\partial : \S \to \S$ are a special case.
\end{example}

\begin{example}
In the subcategory of square matrices, matrix exponential is an example of
bubble for $\S = \R, \C$. When $\S = \B$, square matrices are finite
graphs and reflexive transitive closure is an example.
\end{example}

\begin{example}
Bubbles can encode the standard non-linear operators used in machine learning.
The sigmoid $\sigma(x) = 1 / (1 + e^{-x})$ and rectified linear unit
$\sigma(x) = \max(0, x)$ are pointwise bubbles $\sigma : \R \to \R$.
The softmax function $\sigma : \R^n \to \R^n$ takes a vector $\vec{x}$,
applies exponential pointwise then normalises by $\sum_{i<n} e^{\vec{x}_i}$.
It can be drawn as a bubble around the diagram for the vector $\vec{x}$.
Bubbles may also depend on the labels from the dataset.
Take a loss function such as the relative entropy $l(\vec{y}, \vec{y}^\star)
= \sum_{i < m} \vec{y}_i \log(\vec{y}_i / \vec{y}^\star_i)$.
The partially-applied loss function $l(-, \vec{y}^\star) : \R^m \to \R$ for the
label $\vec{y}^\star \in \R^m$ can be drawn as a bubble around the diagram for
the prediction $\vec{y} \in \R^m$.
\end{example}

Bubbles compose by nesting, this defines a category of post-processes
$pp(\mathbf{C})$. The objects are pairs of objects from $\mathbf{C}$, arrows
$(x, y) \to (x', y')$ are $c$-coloured bubbles such that
$\beta^c_\dom(x) = x'$ and $\beta^c_\cod(y) = y'$.
If we apply this to the category of matrices, $pp(\mathbf{Mat}_\S)$ is the
category of all matrix-valued functions. In particular, this includes
any feed-forward neural networks. Indeed, take
$f = f_n \circ \dots \circ f_1 : \R^a \to \R^b$
where each layer is given by $f_i(\vec{x_i}) = \sigma(W_i \vec{x_i} + \beta_i)$
for the input vector $\vec{x_i} : 1 \to a_i$,
the parametrised weight matrix $W_i : a_i \to b_i$
and bias vector $\beta_i : 1 \to b_i$ in $\mathbf{Mat}_{\R^n \to \R}$
for $n$ the total number of parameters.
Drawing both the layer $f_i$ and the activation $\sigma : \R \to \R$ as bubbles
we get the following definition:
$$\begin{tikzpicture}
	\begin{pgfonlayer}{nodelayer}
		\node [style=none] (156) at (2.25, 0.775) {};
		\node [style=none] (157) at (2.725, 0.775) {};
		\node [style=none] (158) at (2.725, -0.775) {};
		\node [style=none] (159) at (3.225, 0.275) {};
		\node [style=none] (160) at (3.225, -0.325) {};
		\node [style=none] (161) at (1.75, 0.275) {};
		\node [style=none] (162) at (1.75, -0.775) {};
		\node [style=none, font={\scriptsize}] (163) at (1.9, -0.575) {$\sigma$};
		\node [style=none] (6) at (1.25, 0) {$=$};
		\node [style=none] (8) at (-0.25, 0) {};
		\node [style=none] (9) at (-1, 0) {};
		\node [style=none] (10) at (-0.25, 0.5) {};
		\node [style=none] (11) at (-0.25, -0.5) {};
		\node [style=none] (12) at (-0.5, 0) {$\vec{x_i}$};
		\node [style=none] (13) at (0, 0.25) {$a_i$};
		\node [style=none] (14) at (0.75, 0) {};
		\node [style=none] (21) at (0.5, 0.25) {$b_i$};
		\node [style=none] (54) at (2.75, 0) {};
		\node [style=none] (55) at (2, 0) {};
		\node [style=none] (56) at (2.75, 0.5) {};
		\node [style=none] (57) at (2.75, -0.5) {};
		\node [style=none] (58) at (2.5, 0) {$\vec{y_i}$};
		\node [style=none] (59) at (3, 0.25) {$b_i$};
		\node [style=none] (60) at (3.75, 0) {};
		\node [style=none] (61) at (3.25, 0) {};
		\node [style=none] (63) at (1.75, 0) {};
		\node [style=none] (67) at (3.5, 0.25) {$b_i$};
		\node [style=none] (69) at (5, 0) {where};
		\node [style=none] (94) at (7, 0) {};
		\node [style=none] (95) at (6.25, 0) {};
		\node [style=none] (96) at (7, 0.5) {};
		\node [style=none] (97) at (7, -0.5) {};
		\node [style=none] (98) at (6.75, 0) {$\vec{y_i}$};
		\node [style=none] (99) at (7.5, 0) {};
		\node [style=none] (123) at (8.25, 0) {$=$};
		\node [style=none] (124) at (9.75, 0) {};
		\node [style=none] (125) at (9, 0) {};
		\node [style=none] (126) at (9.75, 0.5) {};
		\node [style=none] (127) at (9.75, -0.5) {};
		\node [style=none] (128) at (9.5, 0) {$\beta_i$};
		\node [style=none] (129) at (10.25, 0) {};
		\node [style=none] (130) at (12, 0) {};
		\node [style=none] (131) at (11.25, 0) {};
		\node [style=none] (132) at (12, 0.5) {};
		\node [style=none] (133) at (12, -0.5) {};
		\node [style=none] (134) at (11.75, 0) {$\vec{x_i}$};
		\node [style=none] (135) at (12.25, 0.25) {$a_i$};
		\node [style=none] (136) at (12.5, 0) {};
		\node [style=none] (137) at (12.5, 0.5) {};
		\node [style=none] (138) at (12.5, -0.5) {};
		\node [style=none] (139) at (13.5, 0.5) {};
		\node [style=none] (140) at (13.5, -0.5) {};
		\node [style=none] (141) at (13.5, 0) {};
		\node [style=none] (142) at (13, 0) {$W_i$};
		\node [style=none] (143) at (13.75, 0.25) {$b_i$};
		\node [style=none] (144) at (10.75, 0) {$+$};
		\node [style=none] (145) at (14, 0) {};
		\node [style=none] (146) at (10, 0.25) {$b_i$};
		\node [style=none] (147) at (7.25, 0.25) {$b_i$};
		\node [style=none] (148) at (-0.725, 0.775) {};
		\node [style=none] (149) at (-0.25, 0.775) {};
		\node [style=none] (150) at (-0.25, -0.775) {};
		\node [style=none] (151) at (0.25, 0.275) {};
		\node [style=none] (152) at (0.25, -0.325) {};
		\node [style=none] (153) at (-1.225, 0.275) {};
		\node [style=none] (154) at (-1.225, -0.775) {};
		\node [style=none, font={\scriptsize}] (155) at (-1.075, -0.575) {$f_i$};
	\end{pgfonlayer}
	\begin{pgfonlayer}{edgelayer}
		\draw [style=boxedge] (156.center)
			 to [bend right=45, looseness=1.25] (161.center)
			 to (162.center)
			 to (158.center)
			 to [bend right=45, looseness=1.50] (160.center)
			 to (159.center)
			 to [bend right=45, looseness=1.25] (157.center)
			 to cycle;
		\draw [style=boxedge] (148.center)
			 to [bend right=45, looseness=1.25] (153.center)
			 to (154.center)
			 to (150.center)
			 to [bend right=45, looseness=1.50] (152.center)
			 to (151.center)
			 to [bend right=45, looseness=1.25] (149.center)
			 to cycle;
		\draw (9.center) to (10.center);
		\draw (10.center) to (11.center);
		\draw (11.center) to (9.center);
		\draw (8.center) to (14.center);
		\draw (55.center) to (56.center);
		\draw (56.center) to (57.center);
		\draw (57.center) to (55.center);
		\draw (54.center) to (60.center);
		\draw (95.center) to (96.center);
		\draw (96.center) to (97.center);
		\draw (97.center) to (95.center);
		\draw (94.center) to (99.center);
		\draw (125.center) to (126.center);
		\draw (126.center) to (127.center);
		\draw (127.center) to (125.center);
		\draw (124.center) to (129.center);
		\draw (131.center) to (132.center);
		\draw (132.center) to (133.center);
		\draw (133.center) to (131.center);
		\draw (130.center) to (136.center);
		\draw (137.center) to (138.center);
		\draw (138.center) to (140.center);
		\draw (140.center) to (139.center);
		\draw (139.center) to (137.center);
		\draw (141.center) to (145.center);
	\end{pgfonlayer}
\end{tikzpicture}$$

When the bubble $\beta$ has a derivative $\partial \beta$, we may define the
gradient of bubbled diagrams with the chain rule
$\partial(\beta(f)) = (\partial \beta)(f) \times \partial f$.
In order to make sense of the multiplication, we assume that the homsets of our
category $\mathbf{C}$ have a product on homsets which is compatible with the
sum, i.e. each homset forms a rig\footnote{
We do not assume that products are compatible with composition,
in other words $\mathbf{C}$ need not be rig-enriched.}
and which commutes with the tensor,
i.e. $(f \times f') \otimes (g \times g')
= (f \otimes g) \times (f' \otimes g')$.
The category of matrices $\mathbf{Mat}_\S$ over a rig $\S$ is an example, each
homset $\mathbf{Mat}_\S(m, n)$ is a rig with entrywise sums and products.
Another example is the category of diagrams with spiders on each object, where
the product is given by pre/post-composition with the co/monoid structure.
We get the following equation:
$$
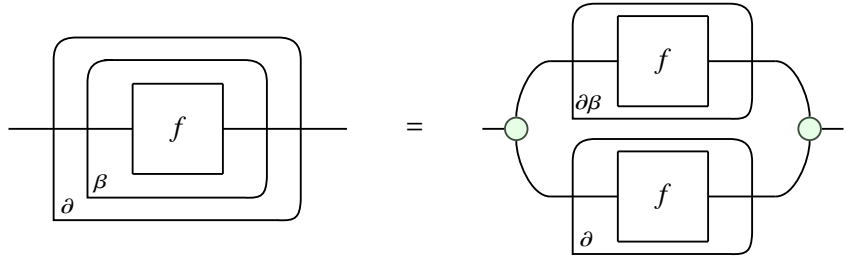
\begin{tikzpicture}[scale=0.6]
	\begin{pgfonlayer}{nodelayer}
		\node [style=none] (86) at (8.475, -2.775) {};
		\node [style=none] (87) at (8.975, -2.325) {};
		\node [style=none] (88) at (5, -2.775) {};
		\node [style=none] (89) at (5.5, -0.225) {};
		\node [style=none] (90) at (8.475, -0.225) {};
		\node [style=none] (91) at (8.975, -0.725) {};
		\node [style=none] (92) at (5, -0.725) {};
		\node [style=none] (63) at (-2.275, -1.525) {};
		\node [style=none] (65) at (-1.775, -1.075) {};
		\node [style=none] (67) at (-5.75, -1.525) {};
		\node [style=none] (61) at (-5.25, 1.525) {};
		\node [style=none] (62) at (-2.275, 1.525) {};
		\node [style=none] (64) at (-1.775, 1.025) {};
		\node [style=none] (66) at (-5.75, 1.025) {};
		\node [style=none, font={\scriptsize}] (68) at (-5.45, -1.2) {$\beta$};
		\node [style=none] (0) at (-2.75, 1) {};
		\node [style=none] (1) at (-4.75, -1) {};
		\node [style=none] (2) at (-2.75, -1) {};
		\node [style=none] (3) at (-4.75, 1) {};
		\node [style=none] (18) at (-3.75, 0) {$f$};
		\node [style=none] (20) at (0, 0) {};
		\node [style=none] (21) at (-7.5, 0) {};
		\node [style=none] (22) at (-2.75, 0) {};
		\node [style=none] (23) at (-4.75, 0) {};
		\node [style=none] (24) at (1.5, 0) {$=$};
		\node [style=none] (25) at (8, 2.5) {};
		\node [style=none] (26) at (6, 0.5) {};
		\node [style=none] (27) at (8, 0.5) {};
		\node [style=none] (28) at (6, 2.5) {};
		\node [style=none] (36) at (7, 1.5) {$f$};
		\node [style=none] (37) at (8, 1.5) {};
		\node [style=none] (38) at (6, 1.5) {};
		\node [style=none] (39) at (8, -0.5) {};
		\node [style=none] (40) at (6, -2.5) {};
		\node [style=none] (41) at (8, -2.5) {};
		\node [style=none] (42) at (6, -0.5) {};
		\node [style=none] (44) at (9, -1.5) {};
		\node [style=none] (50) at (7, -1.5) {$f$};
		\node [style=none] (51) at (8, -1.5) {};
		\node [style=none] (52) at (6, -1.5) {};
		\node [style=none] (53) at (9.5, 1.5) {};
		\node [style=none] (54) at (4.5, 1.5) {};
		\node [style=none] (55) at (4.5, -1.5) {};
		\node [style=none] (56) at (9.5, -1.5) {};
		\node [style=Z] (57) at (10.25, 0) {};
		\node [style=none] (58) at (11, 0) {};
		\node [style=Z] (59) at (3.75, 0) {};
		\node [style=none] (60) at (3, 0) {};
		\node [style=none] (69) at (-1.525, -2.025) {};
		\node [style=none] (70) at (-1.025, -1.575) {};
		\node [style=none] (71) at (-6.5, -2.025) {};
		\node [style=none] (72) at (-6, 2.025) {};
		\node [style=none] (73) at (-1.525, 2.025) {};
		\node [style=none] (74) at (-1.025, 1.525) {};
		\node [style=none] (75) at (-6.5, 1.525) {};
		\node [style=none, font={\scriptsize}] (76) at (-6.2, -1.7) {$\partial$};
		\node [style=none] (77) at (8.475, 0.225) {};
		\node [style=none] (78) at (8.975, 0.675) {};
		\node [style=none] (79) at (5, 0.225) {};
		\node [style=none] (80) at (5.5, 2.775) {};
		\node [style=none] (81) at (8.475, 2.775) {};
		\node [style=none] (82) at (8.975, 2.275) {};
		\node [style=none] (83) at (5, 2.275) {};
		\node [style=none, font={\scriptsize}] (84) at (5.35, 0.55) {$\partial \beta$};
		\node [style=none, font={\scriptsize}] (93) at (5.3, -2.45) {$\partial$};
	\end{pgfonlayer}
	\begin{pgfonlayer}{edgelayer}
		\draw [style=boxedge] (90.center)
			 to (89.center)
			 to [bend right=45, looseness=1.25] (92.center)
			 to (88.center)
			 to (86.center)
			 to [bend right=45, looseness=1.50] (87.center)
			 to (91.center)
			 to [bend right=45, looseness=1.25] cycle;
		\draw [style=boxedge] (81.center)
			 to (80.center)
			 to [bend right=45, looseness=1.25] (83.center)
			 to (79.center)
			 to (77.center)
			 to [bend right=45, looseness=1.50] (78.center)
			 to (82.center)
			 to [bend right=45, looseness=1.25] cycle;
		\draw [style=boxedge] (69.center)
			 to [bend right=45, looseness=1.50] (70.center)
			 to (74.center)
			 to [bend right=45, looseness=1.25] (73.center)
			 to (72.center)
			 to [bend right=45, looseness=1.25] (75.center)
			 to (71.center)
			 to cycle;
		\draw [style=boxedge] (61.center)
			 to [bend right=45, looseness=1.25] (66.center)
			 to (67.center)
			 to (63.center)
			 to [bend right=45, looseness=1.50] (65.center)
			 to (64.center)
			 to [bend right=45, looseness=1.25] (62.center)
			 to cycle;
		\draw (0.center) to (3.center);
		\draw (3.center) to (1.center);
		\draw (1.center) to (2.center);
		\draw (2.center) to (0.center);
		\draw (23.center) to (21.center);
		\draw (22.center) to (20.center);
		\draw (25.center) to (28.center);
		\draw (28.center) to (26.center);
		\draw (26.center) to (27.center);
		\draw (27.center) to (25.center);
		\draw (39.center) to (42.center);
		\draw (42.center) to (40.center);
		\draw (40.center) to (41.center);
		\draw (41.center) to (39.center);
		\draw (37.center) to (53.center);
		\draw (51.center) to (56.center);
		\draw [in=-90, out=0, looseness=0.75] (56.center) to (57);
		\draw [in=90, out=0, looseness=0.75] (53.center) to (57);
		\draw (57) to (58.center);
		\draw (38.center) to (54.center);
		\draw (52.center) to (55.center);
		\draw [in=-90, out=180, looseness=0.75] (55.center) to (59);
		\draw [in=180, out=90, looseness=0.75] (59) to (54.center);
		\draw (59) to (60.center);
	\end{pgfonlayer}
\end{tikzpicture}$$
For scalar diagrams, spiders are empty diagrams and the
equation simplifies to the usual chain rule.

Thus, we can draw both a parametrised quantum circuit and its classical
post-processing as one bubbled diagram in $cq(\mathbf{ZX}_n)$. By applying the
product rule to the quantum circuit and the chain rule to its post-processing,
we can compute a diagram for the overall gradient. This applies to
parametrised quantum circuits seen as machine learning models
\cite{BenedettiEtAl19}, to the patterns of measurement-based quantum
computing seen as ZX-diagrams \cite{DuncanPerdrix10} as well as the quantum
natural language processing of \cite{MeichanetzidisEtAl20a}.


\section*{Conclusion, implementation \& future work}

We introduced diagrammatic differentiation for tensor calculus, using bubbles
to represent the partial derivative of a subdiagram. The product rule
allows to compute the gradient of a diagram from the gradient of its boxes.
Applying this to ZX diagrams, we showed how to compute the gradient of any linear
map with respect to a phase parameter. We then extended this to quantum circuits
with the parameter-shift rule and to neural networks with the chain rule.

Although this work focused on the theoretical foundations of diagrammatic
differentiation, we briefly describe its implementation as
part of the open-source DisCoPy library \cite{DeFeliceEtAl20}. A notebook
with examples is available in the documentation\footnote{
\href{https://discopy.readthedocs.io/en/main/notebooks/diag-diff.html}{
https://discopy.readthedocs.io/en/main/notebooks/diag-diff.html}}.
The \texttt{cqmap} module implements classical-quantum maps as NumPy arrays
\cite{VanDerWaltEtAl11}, with SymPy \cite{MeurerEtAl17} symbols as parameters.
The two modules \texttt{zx} and \texttt{circuit} build upon \texttt{monoidal},
the implementation of diagrams in monoidal categories. They both come with
an \texttt{eval} method which evaluates a diagram as a NumPy array and a
\texttt{grad} method which returns a formal sum of diagrams given a SymPy symbol.
The \texttt{zx} module comes with back-and-forth translations with the PyZX
library \cite{KissingerVanDeWetering19} for automated diagram simplification.
The \texttt{circuit} module interfaces with the tket compiler \cite{SivarajahEtAl20},
allowing to execute the diagrams for circuits and their gradient on quantum hardware.

For now, we have only defined gradients of diagrams with respect to one
parameter at a time. In future work, we plan to extend our definition to
compute the Jacobian of a tensor with respect to a vector of variables.
Other promising directions for research include the study of diagrammatic
differential equations, as well as a definition of integration for diagrams.

\section*{Acknowledgements}

The authors would like to thank the members of the Oxford quantum group for
their insightful feedback.
Special thanks go to Stefano Gogioso for developing the idea of
diagrammatic differentiation with RY during his MSc project \cite{Yeung20}.
We thank the QPL reviewers for their constructive feedback which improved the
presentation of this work.
We also thank Nicola Mariella for his contribution
to the implementation. AT thanks Simon Harrison for the Wolfson Harrison Quantum
Foundation Scholarship.

\bibliographystyle{eptcs}
\bibliography{main}

\end{document}